\documentclass{article}
\usepackage{amsmath, braket, amsthm, amssymb, enumitem, tikz, subcaption}
\usepackage[colorlinks = true]{hyperref}
\usepackage[all]{xy}
\CompileMatrices

\DeclareMathOperator{\tr}{trace}
\DeclareMathOperator{\diag}{diag}
\DeclareMathOperator{\nbd}{nbd}

\newtheorem{theorem}{Theorem}
\newtheorem{corollary}{Corollary}

\newtheorem{example}{Example}
\newtheorem{definition}{Definition}

\title{Quantum discord of states arising from graphs}
\author{Supriyo Dutta\\ Department of Mathematics,\\ Indian Institute of Technology Jodhpur\\ email: \texttt{dutta.1@iitj.ac.in} \vspace{.25cm}\\ 
	Bibhas Adhikari\\ Department of Mathematics,\\ Indian Institute of Technology Kharagpur\\ email: \texttt{bibhas@maths.iitkgp.ernet.in} \vspace{.25 cm}\\
	Subhashish Banerjee\\ Department of Physics,\\Indian Institute of Technology Jodhpur\\ eamil: \texttt{subhashish@iitj.ac.in} }
\date{\today}

\begin{document}
	
	\maketitle
	
	\begin{abstract}
		Quantum discord refers to an important aspect of quantum correlations for bipartite quantum systems. In our earlier works we have shown that corresponding to every graph (combinatorial) there are quantum states whose properties are reflected in the structure of the corresponding graph. Here, we attempt to develop a graph theoretic study of quantum discord that corresponds to a necessary and sufficient condition of zero quantum discord states which says that the blocks of density matrix corresponding to a zero quantum discord state are normal and commute with each other. These blocks have a one to one correspondence with some specific subgraphs of the graph which represents the quantum state. We obtain a number of graph theoretic properties representing normality and commutativity of a set of matrices which are indeed arising from the given graph. Utilizing these properties we define graph theoretic measures for normality and commutativity that results a formulation of graph theoretic quantum discord. We identify classes of quantum states with zero discord using the said formulation. 
	\end{abstract}

	\newpage 
	\section{Introduction}
	Graph theory \cite{west2001introduction} is a well-established branch of mathematics. It has made significant contributions to quantum physics \cite{hall2013quantum} and information theory \cite{berkolaiko2013introduction, dur2003multiparticle, braunstein2006laplacian}. Graphs provide a pictorial representation of quantum states \cite{adhikari2012graph}. They have been used to interpret separability property of quantum states \cite{dutta2016bipartite}, and to model useful unitary operations \cite{dutta2016graph}. Quantum correlations \cite{modi2012classical} are useful resources in quantum information theory \cite{barnett2009quantum}. Important facets of quantum correlations are quantum entanglement \cite{horodecki2009quantum} and quantum discord  \cite{henderson2001classical, ollivier2001quantum, dakic2010necessary, adhikari2012operational}. Here, we attempt to provide a graph theoretical interpretation of quantum discord.
	
	In quantum mechanics, a density matrix $\rho$ is a positive semidefinite, Hermitian matrix with unit trace, acting on a Hilbert space, say $\mathcal{H}^{(A)}$. A measure of `information' contained in the quantum state $\rho$ is the von-Neumann entropy, $S(\rho) = -\tr(\rho \log(\rho))$. A bipartite density matrix acts on a Hilbert space $\mathcal{H}^{(A)} \otimes \mathcal{H}^{(B)}$, where $\otimes$ denotes Kronecker (tensor) product, throughout this article. We denote the reduced density matrix in $\mathcal{H}^{(B)}$ with $\rho_B$. Let $\{\ket{k_B}: k = 1, 2, \dots \}$ be the standard computational basis of the Hilbert space $\mathcal{H}^{(B)}$. The conditional entropy may be defined with $S(A|\{\ket{k_B}\})$, which is given by $\sum_k p_{k_B} S(\rho_{k_B})$ where $\rho_{k_B} = \frac{1}{p_{k_B}} \bra{k_B} \rho \ket{k_B}$, and $p_{k_B} = \tr_A(\bra{k_B} \rho \ket{k_B})$. Further, conditional entropy may be expressed as $S(\rho) - S(\rho_B)$. These two quantities are equal for classical systems but differ for quantum systems. Quantum discord is the difference between two classically equivalent faces of mutual information.
	\begin{definition}{\bf Quantum discord}
		Given a quantum state $\rho$ acting on a bipartite system $\mathcal{H}^{(A)} \otimes \mathcal{H}^{(B)}$, the quantum discord is defined by \cite{huang2011new}
		$$\mathcal{D}_{\{k_B\}}(\rho) = S(A|\{\ket{k_B}\}) - \big[S(\rho) - S(\rho_B)\big].$$
	\end{definition}
	Let $\{\rho^{(a)}_i\}$ be a set of density matrices in $\mathcal{H}^{(A)}$. Quantum discord is zero for pointer states that may be expressed as,
	\begin{equation}
		\rho = \sum_i p_i \rho^{(a)}_i \otimes \ket{k_b}\bra{k_b}.
	\end{equation}

	Understanding the nature of zero quantum discord states is an important stepping stone towards understanding quantum discord, in that it acts as preliminary step to distinguish quantum from the classical aspects. It has been used to understand the completely positive evolution of a system \cite{shabani2009vanishing, rodriguez2008completely, sabapathy2013quantum}, local broadcasting \cite{barnum1996noncommuting, piani2008no}. Thus, finding zero discord quantum states is important in quantum information theory. Corresponding to any graph $G$, there are quantum states $\rho(G)$, defined below. Here, we present a new combinatorial significance to the pointer states. We study a graph theoretic interpretation of binary, normal and commutating matrices. In this framework, we provide a constructive method to generate at least one quantum state with zero discord in an arbitrary dimensional bipartite system. We come up with an idea of graph theoretic measure of quantum discord applicable for quantum states related to graphs. This article contains a considerable study on combinatorial properties of binary matrices along with their physical significance. A number of important quantum states can be represented with wighted graphs. They require a set of additional criteria. This motivates us to study discord of a larger class of quantum states in a forthcoming work \cite{next}.
	 
	This article is organized as follows. In section $2$, we compile a number of nomenclatures and results from graph theory, which would be of use to us subsequently. We generate density matrices corresponding to combinatorial graphs. The combinatorial properties of normal and commutative matrices are investigated in section $3$. These are used to investigate graph theoretic quantum discord in section $4$. We also propose a measure of quantum discord in terms of graph theoretic parameters. Section $5$ is dedicated to graph theoretic quantum states with zero discord.  We then conclude.

	\section{Preliminaries}
	
	In this section we provide a brief review on simple graphs and describe the quantum states arising from them \cite{braunstein2006laplacian, adhikari2012graph}. A simple graph $G = (V(G), E(G))$ consists of a vertex set $V(G)$ and an edge set $E(G)  \subset V(G)\times V(G)$ such that $(i,i)\notin E(G)$ for any $i\in V$ and any edge $(i,j)\in E(G)$ is treated as the same edge $(j,i)\in E(G).$ The number of vertices of $G$ which we denote by $ \#(V(G))$ is called the order of $G.$ In this paper we consider only  finite graphs, that is, $\#(V(G))<\infty.$  The adjacency matrix of a graph $G$ on $N$ vertices is a symmetric binary matrix, that is, a $(0,1)$ matrix $A = (a_{ij})_{N \times N}$ defined as  \cite{bapat2010graphs}
	\begin{equation}
	a_{ij} = \begin{cases} 1 & ~\text{if}~ (i,j) \in E(G), \\ 0 & ~\text{if}~ (i,j) \notin E(G).\end{cases}
	\end{equation}
	The degree of a vertex $i$ is $d_i = \sum_{j = 1}^N a_{ij}$. The degree matrix of the graph $G$ is given by $D(G) = \diag\{d_1, d_2, \dots d_N\}$ and we define the total degree of $G$ as $d = \sum_{i = 1}^N d_i = \tr(D)$. The combinatorial Laplacian matrix and the signless Laplacian matrix associated with the graph $G$ are defined by $$L(G) = D(G) - A(G), \,\, Q(G) = D(G) + A(G),$$ respectively  \cite{bapat2010graphs, cvetkovic2007signless}. Note that $\tr(L(G)) = \tr(Q(G)) = d$ and both $L(G), Q(G)$ are symmetric positive semidefinite matrices. 

	Recall that density matrix representation of a quantum state is described by a Hermitian positive semidefinite matrix with unit trace \cite{barnett2009quantum}. Thus density matrices corresponding to a graph $G$ are defined by
	\begin{equation}\label{def:densitym}
	\rho_l(G) = \frac{1}{d} L(G) ~\text{and}~ \rho_q(G) = \frac{1}{d} Q(G).
	\end{equation}
	They were introduced in  \cite{braunstein2006laplacian, adhikari2012graph} and represent quantum states of dimension $N$. We denote $\rho_l(G)$ and $\rho_q(G)$ together by $\rho(G)$ when no confusion arises. It is important to note that $L(G)$ and $Q(G)$ depend on the vertex labellings, and hence different labellings on the vertex set of a graph generate different quantum states \cite{braunstein2006laplacian, adhikari2012graph}.

	Given a graph $G$ on $N=mn$ vertices the vertex set $V(G)$ can be partitioned into $m$ classes, say  $C_1, C_2, \dots C_m$ such that 
	\begin{equation}\label{clustering}
	\begin{split}
	& V = C_1 \cup C_2 \cup \dots \cup C_m\\
	& C_\mu \cap C_\nu = \emptyset ~\text{for}~ \mu \neq \nu ~\text{and}~ \mu, \nu = 1, 2, \dots m\\
	& C_\mu = \{v_{\mu 1}, v_{\mu 2}, \dots v_{\mu n}\}.
	\end{split}
	\end{equation} The induced subgraph of $G$ defined by $C_\mu$, that is the graph with vertex set $C_\mu$ and edge set $\{(i,j )\in E(G) : i, j \in C_\mu \}$ is called a cluster in $G$ and we denote it by $\langle C_\mu\rangle.$ We denote the bipartite graph defined by a pair $C_\mu, C_\nu, \mu\neq \nu$ consisting of the edge set $\{(i,j) \in E(G): i \in C_\mu, j\in C_\nu\}$ and vertex set $C_\mu\cup C_\nu$ as $\langle C_\mu, C_\nu\rangle$ which is a subgraph of $G$ representing the edges between the pair of clusters for any $\mu,\nu=1,\hdots,m$. Hence, the adjacency matrix associated with $G$ can be represented as the block matrix
	\begin{equation}\label{adjacency}
		A(G) = \begin{bmatrix} A_{11} & A_{12} & \dots & A_{1m} \\ A_{21} & A_{22} & \dots & A_{2m} \\ \vdots & \vdots & \ddots & \vdots \\ A_{m1} & A_{m2} & \dots & A_{mm}\end{bmatrix}_{mn \times mn},
	\end{equation} where $A_{\mu\mu}$ denotes the adjacency matrix of the cluster $C_\mu$ and  \begin{equation}\label{adj2}
		\begin{bmatrix} 0 & A_{\mu \nu} \\ A_{\nu \mu} & 0\end{bmatrix} = \begin{bmatrix} 0 & A_{\mu \nu} \\ A_{\mu \nu}^t & 0\end{bmatrix}
	\end{equation} denotes the adjacency matrix associated with the bipartite graph $\langle C_\mu,C_\nu\rangle$ \cite{dutta2016bipartite}.

	Consequently, the density matrices corresponding to the graph $G$ are block matrices $\rho(G)=[\rho_{\mu\nu}]$ such that
	\begin{equation}\label{block_and_graph}
		\rho_{\mu \nu} = \begin{cases}
		\dfrac{s}{d}A_{\mu \nu} & ~\text{if}~ \mu \neq \nu \\\\
		\dfrac{1}{d}(D_\mu + sA_{\mu \mu}) & ~\text{if}~ \mu = \nu \\\\
	\end{cases}
	\end{equation} where  $D = \diag\{D_1, D_2, \dots D_m\}, D_\mu$ is a diagonal matrix whose diagonal entries are the degrees of the vertices belong to $C_\mu$, $s=1$ if $\rho(G)= \rho_q(G),$ and $s=-1$ if $\rho(G)= \rho_l(G).$ Thus, $\rho(G)$ represents quantum states corresponding to a bipartite system of order $m\times n$. Finally we conclude the section with the following definition which will be used later.
		
	\begin{definition}{\bf Edge characteristic function:}\label{ecf}
		Given a graph $G,$ the function $\mathcal{X}: V(G) \times V(G) \rightarrow \{0, 1\}$ defined by
		$$\mathcal{X}(v_{\mu, i},v_{\nu, j}) \equiv \mathcal{X}_{\mu, \nu}(i,j)= \begin{cases} 1 & \text{if}~ (v_{\mu, i},v_{\nu, j}) \in E(G), \\ 0 & \text{if}~ (v_{\mu, i},v_{\nu, j}) \notin E(G),\end{cases}$$ for all $\mu,\nu=1,\hdots,m$ and $i,j=1,\hdots,n$ is called an edge characteristic function.
	\end{definition}

	\section{Graph theoretic interpretations of normal and commuting matrices}
	
	As mentioned earlier, the zero quantum discord states are given by the normal and commuting blocks of the corresponding density matrices \cite{huang2011new}. In this section, we determine the structural properties of a graph on $mn$ vertices such that its density matrix has blocks that are normal and commute pairwise. We derive properties of the clusters $\langle C_\mu\rangle$ and the bipartite graphs $\langle C_\mu,C_\nu\rangle$ such that the matrices $\rho_{\mu\nu},  \mu,\nu=1,\hdots,m$ form a set of normal commutating matrices. 

	First, we discuss some notations and observations regarding graphs generated from a binary matrix. In what follows, given a vector $a=[a_1 \, a_2 \, \hdots a_n]^t\in\{0,1\}^n,$ we denote $$\overline{a}=\{i : a_i=1, 1\leq i\leq n\}.$$ Hence, given $a,b\in\{0,1\}^n$ we obtain \begin{equation}\label{inner_product} a^tb=\#(\overline{a}\cap \overline{b}).\end{equation}

	For a matrix $M = [m_{ij}]\in\{0,1\}^{n \times n}$, we denote $m_{i*}$ and $m_{*j}$ as the $i$-th row and $j$-th column vectors, respectively. Corresponding to every such matrix $M$ there is a simple bipartite graph $G_M = (V(G_M), E(G_M))$ of order $2n$ whose adjacency matrix is given by
	\begin{equation}\label{matrix_to_bipartite}
	A(G_M)=\mathcal{M} = \begin{bmatrix} 0 & M \\ M^t & 0 \end{bmatrix}.
	\end{equation} We mention that corresponding to any non-negative matrix, that is, a matrix whose all the entries are non-negative such a bipartite graph can also be defined, for example, see \cite{brualdi2006combinatorial}.

	Assuming the partitioned vertex sets of $V(G_M)$ as $C_\mu = \{v_{\mu,1}, v_{\mu,2}, \dots v_{\mu,n}\}$ and $C_\nu = \{v_{\nu,1}, v_{\nu,2}, \dots v_{\nu,n}\},$ note that $(v_{\mu i}, v_{\nu j}) \in E(G_M)$ if and only if $m_{ij} = 1$. Thus, $G_M = \langle C_\mu, C_\nu \rangle$. As $G_M$ is bipartite, the neighborhood of a vertex $v_{\mu i}$ in $G_M$ is given by
	\begin{equation}\label{mv}
		 \nbd_\mathcal{M}(v_{\mu i}) = \{v_{\nu j}: (v_{\mu i}, v_{\nu j}) \in E(G)\}\subseteq C_\nu.
	\end{equation}  
	Similarly, $ \nbd_\mathcal{M}(v_{\nu i}) \subseteq C_\mu$. Now we define a set of numbers for any $v_{\mu i}\in C_\mu$ and $v_{\nu j}\in C_\nu, 1\leq \mu,\nu \leq m$ corresponding to the bipartite graph $\langle C_\mu,C_\nu\rangle$ with the help of (\ref{mv}) as 
	\begin{eqnarray} 
	\nbd(v_{\mu i}) &=& \{j : v_{\nu j}\in  \nbd_\mathcal{M}(v_{\mu i}) \} = \overline{m_{i*}^t}, \\  \nbd(v_{\nu i}) &=& \{j : v_{\mu j}\in  \nbd_\mathcal{M}(v_{\nu i}) \} = \overline{m_{*i}}, 
	\end{eqnarray} 
	which are extensively used in the sequel.

	Let $0_{1,n}$ and $0_{n,1}$ be the zero row and column vectors. Note that, the $i$-th row of $\mathcal{M}$, that is $[0_{1,n} \, m_{i*}]\in\{0,1\}^{2n}$ depicts edges incident to $v_{\mu i}, 1\leq i\leq n$, and hence 
	$$\overline{[0_{1,n} \, m_{i*}]^t} = \overline{m_{i*}^t}= \nbd_\mathcal{M}(v_{\mu i}).$$ 
	Similarly, the $(n + i)$-th column of $\mathcal{M}$, that is $ \left[\begin{matrix} m_{*i}\\ 0_{n,1}\end{matrix}\right]$  represents the edges incident to $v_{\nu i}$ and thus 
	$$\overline{ \left[\begin{matrix} m_{*i}\\ 0_{n,1}\end{matrix}\right]} = \overline{m_{*i}}= \nbd_\mathcal{M}(v_{\nu i}).$$

	In particular, any  symmetric matrix $M\in\{0,1\}^{n\times n}$ with diagonal entries zero can be considered as an adjacency matrix of a graph $G(M)$. Let $V(G(M)) = C_\mu = \{v_{\mu,1}, v_{\mu,2}, \dots v_{\mu,n}\}$. Then  $(v_{\mu,i}, v_{\mu,j}) \in E(G(M))$ if and only if $m_{ij} =1$. Thus $G(M) = \langle C_\mu \rangle$. 

	We illustrate the above discussion using the following example.

	\begin{example}
		Consider the matrix $M = \begin{bmatrix} 0 & 1 & 1 \\ 1 & 0 & 0 \\ 1 & 0 & 0 \end{bmatrix}$. The corresponding bipartite graph, $G_M$ is:
		$$
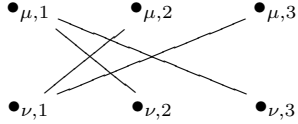
\xymatrix{\bullet_{\mu,1} \ar@{-}[dr] \ar@{-}[drr] & \bullet_{\mu, 2} \ar@{-}[dl] & \bullet_{\mu,3} \ar@{-}[dll] \\ \bullet_{\nu, 1} & \bullet_{\nu, 2} & \bullet_{\nu, 3} }$$
		Consider, $m_{*1} = (0, 1, 1)^t$, that is $\overline{m_{*1}} = \{2, 3\}$. Note that, $\nbd_{\mathcal{A}}(v_{\nu 1}) = \{v_{\mu,2}, v_{\mu,3}\}$. Also, $M$ is a symmetric binary matrix with zero diagonal entries. Thus $M$ is the adjacency matrix of the following graph $G(M)$
		$$\xymatrix{\bullet_{\mu 3} \ar@{-}[r] & \bullet_{\mu 1} \ar@{-}[r] & \bullet_{\mu 2}}.$$
	\end{example}
			
	We characterize commutativity of two binary matrices in the following results  by using the bipartite graphs introduced above. Next, we also provide a measure of non-commutativity of two binary matrices using these results.
	
	\begin{theorem}\label{commutativity}
		Let the bipartite graphs corresponding to the matrices $A, B\in\{0,1\}^{n\times n}$ be $G_A = \langle C_\mu, C_\nu \rangle$ and $G_B = \langle C_\alpha, C_\beta \rangle$, respectively. Then $AB=BA$ if and only if for all $i, j$ with $1 \le i,j \le n$,
		$$\#(\nbd(v_{\mu i}) \cap \nbd(v_{\beta j})) = \#(\nbd(v_{\nu j}) \cap \nbd(v_{\alpha i})).$$
	\end{theorem}
	
	\begin{proof} 
		Note that $AB = BA$ holds if and only if $(AB)_{ij} = (BA)_{ij}$ for all $i, j$ with $1 \le i,j \le n$. Now applying equation (\ref{inner_product}) we get,
		\begin{equation}\label{commutativity_main}
			\begin{split} 
				(AB)_{ij} & = \sum_{k = 1}^n a_{ik}b_{kj} =  a_{i*}^t b_{*j}  = \#(\nbd(v_{\mu i}) \cap \nbd(v_{\beta j})),\\
				(BA)_{ij} & = \sum_{k = 1}^n b_{ik}a_{kj} =  b_{i*}^t a_{*j} = \#(\nbd(v_{\nu j}) \cap \nbd(v_{\alpha i})).
			\end{split}
		\end{equation}
	Hence the desired result follows.\end{proof}
	
	Obviously if $AB\neq BA$ the corresponding condition on graphs do not hold. The non-commutativity of $A$ and $B$ is captured in $E(G_A)$, and $E(G_B)$. Hence we introduce the following quantity as a measure of non-commutativity of any two matrices $A, B\in\{0,1\}^{n\times n}.$
	\begin{equation}\label{commutativity_measure_1}
	\mathcal{NC}_1(A,B)= \sum_{i,j}  \big| \#(\nbd(v_{\mu i}) \cap \nbd(v_{\beta j})) - \#(\nbd(v_{\nu j}) \cap \nbd(v_{\alpha i}))\big|.
	\end{equation} 
	
	\begin{corollary}\label{commutativity_1}
		Let $A=[a_{ij}]\in\{0,1\}^{n\times n}$ be a symmetric matrix with diagonal entries zero and $B=[b_{ij}]\in\{0,1\}^{n\times n}.$ Assume that $G(A) = \langle C_\mu \rangle$ and $G_B = \langle C_\alpha, C_\beta \rangle$ are the graphs corresponding to $A$ and $B$, respectively. Then $AB=BA$ if and only if for all $i, j$ with $1 \le i,j \le n$,
		$$\#(\nbd(v_{\mu i}) \cap \nbd(v_{\beta j})) = \#(\nbd(v_{\mu j}) \cap \nbd(v_{\alpha i})).$$
	\end{corollary}
	
	\begin{proof}
		We have already justified that, $\overline{a_{i*}^t} = \nbd(v_{\mu i}) = \overline{a_{*i}}$, for all $i = 1, 2, \dots n$. Further $AB=BA$ if and only if $ a_{i*}^t b_{*j} = b_{i*}^t a_{*j}$ for all $i,j.$ Applying the symmetry of $A$, we obtain $ a_{i*}^t b_{*j} = a_{j*}^t b_{i*}$. Using the graph theoretic convention $\#(\nbd(v_{\mu i}) \cap \nbd(v_{\beta j})) = \#(\nbd(v_{\mu j}) \cap \nbd(v_{\alpha i}))$.
	\end{proof}

	When such matrices $A,B$ in the above corollary do not commute we denote \begin{equation}\label{measure:nc2ij}
		\mathcal{N}\mathcal{C}_2(A,B)_{ij}=  \#(\nbd(v_{\mu i}) \cap \nbd(v_{\beta j})) - \#(\nbd(v_{\mu j}) \cap \nbd(v_{\alpha i})),
	\end{equation} and define a measure of non-commutativity of the pair of matrices $A,B$ as 
	\begin{equation}\label{measure:nc2}
		\mathcal{N}\mathcal{C}_2(A,B)= \sum_{i,j} \big| \#(\nbd(v_{\mu i}) \cap \nbd(v_{\beta j})) - \#(\nbd(v_{\mu j}) \cap \nbd(v_{\alpha i}))\big|.
	\end{equation}
	
	\begin{corollary}\label{commutativity_2}
		Let $A=[a_{ij}], B=[b_{ij}]\in\{0,1\}^{n\times n}$ be symmetric matrices with zero diagonal entries. Suppose $G(A) = \langle C_\mu \rangle$ and $G(B) = \langle C_\nu \rangle.$ Then $AB=BA$ if and only if for every $i, j$ with $1 \le i, j \le n$,
		$$\#(\nbd(v_{\mu i}) \cap \nbd(v_{\nu j})) = \#(\nbd(v_{\mu j}) \cap \nbd(v_{\nu i})).$$
	\end{corollary}

	\begin{proof}
		The proof follows from the above corollary by setting $\alpha = \beta = \nu$.
	\end{proof}

	Thus, given two symmetric binary matrices with diagonal entries zero we denote 
	\begin{equation}\label{measure:nc3ij}
		\mathcal{N}\mathcal{C}_3(A,B)_{ij}=\#(\nbd(v_{\mu i}) \cap \nbd(v_{\nu j})) - \#(\nbd(v_{\mu j}) \cap \nbd(v_{\nu i})),
	\end{equation} and define a measure of non-commutativity of $A$ and $B$ as 
	\begin{equation}\label{measure:nc3}
		\mathcal{N}\mathcal{C}_3(A,B)=\sum_{i,j} \big|\#(\nbd(v_{\mu i}) \cap \nbd(v_{\nu j})) - \#(\nbd(v_{\mu j}) \cap \nbd(v_{\nu i}))\big|.
	\end{equation}
	
	Now we provide graph theoretic interpretation of normality of a binary matrix as follows. 

	\begin{theorem}\label{normality}
		Let $A=[a_{ij}]\in\{0,1\}^{n\times n}$ and $G_A = \langle C_\mu, C_{\nu} \rangle$ be the bipartite graph corresponding to $A.$ Then $A$ is normal, that is, $AA^t=A^tA$ if and only if for every $i$ and $j$ with $1 \le i, j \le n$,
		$$\#(\nbd(v_{\mu i}) \cap \nbd(v_{\mu j})) = \#(\nbd(v_{\nu i}) \cap \nbd(v_{\nu j})).$$
	\end{theorem}	
	
	\begin{proof}
		Let $B = (b_{ij})_{n \times n} = (a_{ji})_{n \times n} = A^t$. Clearly, $b_{i*} = a_{*i}$ and $b_{*i} = a_{i*}$ for all $i$. Note that,
		\begin{equation} 
			(AA^t)_{ij} = \sum_{k = 1}^n a_{ik}b_{kj} =  a_{i*}^t b_{*j} = a_{i*}^t a_{j*} \rangle = \#(\nbd(v_{\mu i}) \cap \nbd(v_{\mu j})).
		\end{equation} 
		Similarly, $(A^tA)_{ij} = \#(\nbd(v_{\nu i}) \cap \nbd(v_{\nu j}))$. Hence, for any two $i$, and $j$ with $1 \le i,j \le n$ we have, $\#(\nbd(v_{\mu i}) \cap \nbd(v_{\mu j})) = \#(\nbd(v_{\nu i}) \cap \nbd(v_{\nu j}))$.
	\end{proof}
	
	When $A\in\{0,1\}^{n\times n}$ is not a normal matrix we define its non-normality in terms of the edges in $G_A$ by the following quantity:
	\begin{equation}\label{normality_measure_1}
		\mathcal{NN}(A)=\sum_{i,j} |\#(\nbd(v_{\mu i}) \cap \nbd(v_{\mu j})) - \#(\nbd(v_{\nu i}) \cap \nbd(v_{\nu j})) |.
	\end{equation}

	\section{Quantum discord of states corresponding to graphs}
	
	We first recall the clusters for a given graph $G$ on $mn$ vertices that are mentioned in equation (\ref{clustering}). Note that any such simple graph $G$ can be partitioned into edge-disjoint subgraphs $\langle C_\mu \rangle$ and $\langle C_\mu, C_\nu \rangle$,  $1 \le \mu, \nu \le m$, and properties of these  subgraphs determine some properties of $G$. In order to determine the zero quantum discord states which arise from $G$, the blocks of $\rho(G)=[\rho_{\mu\nu}]$ must satisfy the following conditions: 
	\begin{eqnarray} 
	\rho_{\mu\mu}^t\rho_{\mu\mu} &=& \rho_{\mu\mu}\rho_{\mu\mu}^t \label{prop1},\\ 
	\rho_{\mu\nu}^t\rho_{\mu\nu} &=& \rho_{\mu\nu}\rho_{\mu\nu}^t, \mu\neq \nu \label{prop2} \\ \rho_{\mu\nu}\rho_{\alpha\beta} &=& \rho_{\alpha\beta}\rho_{\mu\nu}, \mu\neq \nu, \alpha\neq \beta, (\mu,\nu)\neq (\alpha,\beta) \label{prop3} \\
	\rho_{\mu\mu}\rho_{\alpha\beta} &=& \rho_{\alpha\beta}\rho_{\mu\mu}, \alpha\neq \beta \label{prop4}\\ \rho_{\mu\mu}\rho_{\nu\nu} &=& \rho_{\nu\nu}\rho_{\mu\mu}, \label{prop5}\end{eqnarray} 
	where $\rho_{\mu\mu}=D_\mu+ A_{\mu\mu}$ if $\rho(G)=\rho_l(G),$  $\rho_{\mu\mu}=D_\mu- A_{\mu\mu}$ if $\rho(G)=\rho_s(G),$ $\rho_{\mu\nu}=A_{\mu\nu}, \mu\neq \nu$, and  $A_{\mu\mu}, A_{\mu\nu}$ are described in equations (\ref{adjacency}) and (\ref{adj2}). Thus the quantum states $\rho(G)$ arising from a graph $G$ must satisfy the conditions (\ref{prop1})-(\ref{prop5}) to represent a state with zero quantum discord. It is needless to mention that the structural properties of the clusters and the edges between the clusters determine the same. 

	Observe that the condition (\ref{prop1}) satisfies trivially for any graph $G$ since the matrix $A_{\mu\mu}$ represents a symmetric adjacency matrix associated with the cluster $\langle C_\mu\rangle$. Also, $D_\mu$ is a diagonal matrix. The condition (\ref{prop2}) is satisfied by $G$ if all the bipartite graphs $\langle C_\mu, C_\nu \rangle$ meet the normality condition given in Theorem \ref{normality}. If there are some block matrices $\rho_{\mu \nu}$ which are not normal, hence violate (\ref{prop2}), the amount of non-normality can be measured by using formula (\ref{normality_measure_1}) considering all pairs of $\mu, \nu$ such that $\mu\neq \nu.$ Thus the quantity  
	\begin{equation}\label{normality_deviation}
		\sum_{\mu\neq \nu} \mathcal{NN}(A_{\mu\nu})
	\end{equation} measures the violation of (\ref{prop2}).

	The condition (\ref{prop3}) is satisfied if all the pair of bipartite graphs $\langle C_\mu, C_\nu \rangle$, and $\langle C_\alpha, C_\beta \rangle, 1\leq \mu,\nu,\alpha,\beta\leq m, \mu\neq \nu, \alpha\neq \beta, (\mu,\nu)\neq (\alpha,\beta)$ satisfy Theorem \ref{commutativity}. If the condition gets vilolated, the amount of non-commutativity defined in (\ref{commutativity_measure_1}) can be used to measure the violation of condition (\ref{prop3}) due to all pairs of $A_{\mu\nu}, A_{\alpha\beta}$ as  \begin{equation}\label{commutativity_deviation}
		\sum_{\mu\neq\nu, \alpha\neq\beta}\mathcal{NC}_1(A_{\mu\nu}, A_{\alpha\beta}).
	\end{equation}

	Note that the condition (\ref{prop4}) deals with the commutativity between $\rho_{\mu \mu}$ and $\rho_{\alpha \beta}, \alpha\neq \beta$, that is, the graph $G$ will satisfy $\rho_{\mu \mu} \rho_{\alpha \beta} = \rho_{\alpha \beta} \rho_{\mu \mu}, 1\leq \mu,\alpha,\beta\leq m$. Thus
	\begin{equation}
		\begin{split} 
			& \frac{1}{d}(D_{\mu} +s A_{\mu \mu})\frac{s}{d}A_{\alpha \beta} = \frac{s}{d}A_{\alpha \beta} \frac{1}{d}(D_{\mu} +s A_{\mu \mu})\\
			\Rightarrow~ & D_{\mu}A_{\alpha \beta} +s A_{\mu \mu}A_{\alpha \beta} = A_{\alpha \beta}D_{\mu} +s A_{\alpha \beta}A_{\mu \mu}
		\end{split} 
	\end{equation} where $s=1$ if $\rho(G)=\rho_l(G)$ and $s=-1$ if $\rho(G)=\rho_s(G).$
	Rearranging the terms we obtain
	\begin{equation}
		(D_{\mu}A_{\alpha \beta} - A_{\alpha \beta}D_{\mu}) + s (A_{\mu \mu}A_{\alpha \beta} - A_{\alpha \beta}A_{\mu \mu}) =  0.
	\end{equation}
	Recall that $D_\mu$ represents the diagonal matrix having the deagonal entries as the degrees of the vertices belong to $C_\mu$ and $A_{\mu \mu}$ is the adjacency matrix of the cluster $\langle C_\mu \rangle.$ Besides, $A_{\alpha \beta}$ corresponds to the bipartite graph $\langle C_\alpha, C_\beta \rangle$. Thus the above equation holds if for all $i,j$ with $1 \le i,j \le n$
	\begin{equation}
	\begin{split} 
		& (D_{\mu}A_{\alpha \beta})_{ij} - (A_{\alpha \beta}D_{\mu})_{ij} +s \{ (A_{\mu \mu}A_{\alpha \beta})_{ij} - (A_{\alpha \beta}A_{\mu \mu})_{ij} \} =  0\\
		\Rightarrow~ & d_{\mu i}(A_{\alpha \beta})_{ij} - (A_{\alpha \beta})_{ij}d_{\mu j} +s \{ (A_{\mu \mu}A_{\alpha \beta})_{ij} - (A_{\alpha \beta}A_{\mu \mu})_{ij} \} =  0.
	\end{split}
	\end{equation} 
	Further, $(A_{\alpha \beta})_{ij}$ is either $0$ or $1$ depending on the existence of the edge $(v_{\alpha i}, v_{\beta j})$ in $G$. Thus the graph $G$ satisfies condition (\ref{prop4}) if 
	\begin{equation}
		\mathcal{X}_{\alpha \beta}(i,j)(d_{\mu i} - d_{\mu j}) +s (\#(\nbd(v_{\mu i}) \cap \nbd(v_{\beta j})) - \#(\nbd(v_{\mu j}) \cap \nbd(v_{\alpha i})))= 0
	\end{equation} as follows from Corollary \ref{commutativity_1} for all $1\leq \mu,\alpha,\beta\leq m, \alpha\neq \beta,$ where $\mathcal{X}_{\alpha\beta}$ denotes the edge characteristic function defined in Definition \ref{ecf}. Moreover the violation of the condition (\ref{prop4}) can be represented by 
	\begin{equation}\label{commutativity_1_deviation}
		\sum_{\mu,\alpha\neq\beta} \sum_{i,j} \big|\mathcal{X}_{\alpha \beta}(i,j)(d_{\mu i} - d_{\mu j}) +s\mathcal{NC}_2(A_{\mu\mu}, A_{\alpha\beta})_{ij}\big|
	\end{equation}  where $\mathcal{NC}_2(A_{\mu\mu}, A_{\alpha\beta})_{ij}$ is given by equation (\ref{measure:nc2ij}), $s=1$ if $\rho(G)=\rho_l(G)$ and $s=-1$ if $\rho(G)=\rho_s(G).$ 

	Finally the condition (\ref{prop5}) holds if $\rho_{\mu \mu}\rho_{\nu \nu} = \rho_{\nu \nu}\rho_{\mu \mu}$ which implies that
	\begin{equation}
		\begin{split} 
			& \frac{1}{d}(D_\mu +s A_{\mu \mu}) \frac{1}{d}(D_\nu +s A_{\nu \nu}) = \frac{1}{d}(D_\nu +s A_{\nu \nu}) \frac{1}{d}(D_\mu +s A_{\mu \mu})\\
			\Rightarrow~ & D_\mu D_\nu +s D_\mu A_{\nu \nu} +s A_{\mu \mu}D_\nu + A_{\mu \mu}A_{\nu \nu} = D_\nu D_\mu +s D_\nu A_{\mu \mu} +s A_{\nu \nu}D_\mu + A_{\nu \nu}A_{\mu \mu}\\
			\Rightarrow~ & (A_{\mu \mu}A_{\nu \nu} - A_{\nu \nu}A_{\mu \mu}) +s (D_\mu A_{\nu \nu} - A_{\nu \nu}D_\mu) +s (A_{\mu \mu}D_\nu - D_\nu A_{\mu \mu}) = 0\\
			\Rightarrow~ & (A_{\mu \mu}A_{\nu \nu} - A_{\nu \nu}A_{\mu \mu})_{ij} +s (D_\mu A_{\nu \nu} - A_{\nu \nu}D_\mu)_{ij} +s (A_{\mu \mu}D_\nu - D_\nu A_{\mu \mu})_{ij} = 0,
		\end{split}
	\end{equation}
	holds for all $1 \le i,j \le n$. Note that $A_{\mu \mu}$ and $A_{\nu \nu}$ represent the adjacency matrices corresponding to the clusters $\langle C_\mu\rangle$ and $\langle C_\nu\rangle$ respectively and commutativity of such matrices has been discussed in Corollary \ref{commutativity_2}. Note that \begin{align}
		& (D_\mu A_{\nu \nu} - A_{\nu \nu}D_\mu)_{ij} = d_{\mu i}(A_{\nu \nu})_{ij} - (A_{\nu \nu})_{ij} d_{\mu j} = \mathcal{X}_{\nu \nu}(i,j)(d_{\mu i} - d_{\mu j})\label{eqn:1}\\
		& (A_{\mu \mu}D_\nu - D_\nu A_{\mu \mu})_{ij} = (A_{\mu \mu})_{ij}d_{\nu j} - d_{\nu i}(A_{\nu \nu})_{ij} = \mathcal{X}_{\mu \mu}(i,j)(d_{\nu j} - d_{\nu i}) \label{eqn:2}
	\end{align} 
	which follows from the definition of the edge characteristic function $\mathcal{X},$ and 
	\begin{equation}\label{eqn:3}
		(A_{\mu \mu}A_{\nu \nu} - A_{\nu \nu}A_{\mu \mu})_{ij} = \#(\nbd(v_{\mu i}) \cap \nbd(v_{\nu j})) - \#(\nbd(v_{\mu j}) \cap \nbd(v_{\nu i})).
	\end{equation} Combining the equations (\ref{eqn:1})-(\ref{eqn:3}) we obtain 
	\begin{eqnarray}
		&& \big[\#(\nbd(v_{\mu i}) \cap \nbd(v_{\nu j})) - \#(\nbd(v_{\mu j}) \cap \nbd(v_{\nu i})) \big] \nonumber \\ && +s \big[\mathcal{X}_{\nu \nu}(i,j)(d_{\mu i} - d_{\mu j})\big] +s \big[\mathcal{X}_{\mu \mu}(i,j)(d_{\nu j} - d_{\nu i})\big] = 0
	\end{eqnarray} which must be  satisfied for all $1\leq i,j\leq n$ in order to satisy conditon (\ref{prop5}).

	Further, observe that if the vertices belong to $C_\mu, \mu = 1, 2, \dots m$ have equal degree, then $d_{\mu i} - d_{\mu j} = 0$ as well as $d_{\nu j} - d_{\nu i} = 0$. Then $G$ satisfies condition (\ref{prop5}) if and only if for any two subgraphs $\langle C_\mu \rangle$ and $\langle C_\nu \rangle$, conditions of Corollary \ref{commutativity_2} is fulfilled. Thus a measure of violation of the (\ref{prop5}) can be defined by
	\begin{equation}\label{commutativity_2_deviation}
		\sum_{\mu\neq \nu} \sum_{i,j} \big| \mathcal{NC}_3 (A_{\mu\mu}, A_{\nu\nu})_{ij} +s \big[\mathcal{X}_{\nu \nu}(i,j)(d_{\mu i} - d_{\mu j}) + \mathcal{X}_{\mu \mu}(i,j)(d_{\nu j} - d_{\nu i})\big] \big|
	\end{equation} 
	where $\mathcal{NC}_3 (A_{\mu\mu}, A_{\nu\nu})_{ij}$ is given by (\ref{measure:nc3ij}), $s=1$ if $\rho(G)=\rho_l(G)$ and $s=-1$ if $\rho(G)=\rho_s(G).$

	Based on the discussions above it is obvious that given a graph $G$ on $mn$ vertices with a labelling on the vertices and clusters $\langle C_\mu\rangle, 1\leq \mu\leq m$ a notion of quantum discord for the quantum states $\rho(G)$ can be defined by using the equations  (\ref{commutativity_deviation}),  (\ref{commutativity_deviation}), (\ref{commutativity_1_deviation}) and (\ref{commutativity_2_deviation}).  This definition of quantum discord would then be philosophically different compared to the existing measures of quantum discord, for example see \cite{guo2016non}, as it depends on the structural properties of the clusters in the graph. Thus we introduce the following definition of quantum discord for states arising from a graph.

	\begin{definition}{\bf Graph theoretic quantum discord}\label{gtqc}
		Let $G$ be a graph on $mn$ vertices and $\langle C_\mu\rangle, C_\mu=\{v_{\mu i} : 1\leq i\leq n\}, 1\leq \mu\leq m$ be the clusters in $G.$ Then the quantum discord of the states $\rho(G)=\frac{1}{d}[D(G)+sA(G)], s\in\{1,-1\}$ is given by \begin{eqnarray} \mathcal{QD}(G) &=& \sum_{\mu\neq \nu} \mathcal{NN}(A_{\mu\nu}) + \sum_{\mu\neq \nu} \mathcal{NN}(A_{\mu\nu}) \nonumber \\ && + \sum_{\mu,\alpha\neq\beta} \sum_{i,j} \big|\mathcal{X}_{\alpha \beta}(i,j)(d_{\mu i} - d_{\mu j}) +s\mathcal{NC}_2(A_{\mu\mu}, A_{\alpha\beta})_{ij}\big| \nonumber \\ &&+ \sum_{\mu\neq \nu} \sum_{i,j} \big| \mathcal{NC}_3 (A_{\mu\mu}, A_{\nu\nu})_{ij} +s \big[\mathcal{X}_{\nu \nu}(i,j)(d_{\mu i} - d_{\mu j}) + \mathcal{X}_{\mu \mu}(i,j)(d_{\nu j} - d_{\nu i})\big] \big| \nonumber \end{eqnarray} where $1\leq \nu\leq m,$ $1\leq j\leq n,$ $A(G)=[A_{\mu\nu}], d_{\mu i}$ is the degree of $v_{\mu i},$  and $d$ is the total degree of $G.$
	\end{definition}

	We mention that different labellings on the vertices of the same graph $G$ can produce different states with different quantum discord. The definition helps to generate zero quantum discord states defined by graphs. Indeed, a procedure to create such states with a given dimension $m\times n$ would be to define the edges of the graph such that the quantities defined in equations  (\ref{normality_deviation}),  (\ref{commutativity_deviation}), (\ref{commutativity_1_deviation}) and (\ref{commutativity_2_deviation}) become zero. 

	Now we discuss whether $\mathcal{QD}(G)$ is a valid measure of quantum correlation for bipartite states represented by the density matrices $\rho(G)$  of order $mn.$ First we recall that a general measure $\mathcal{M}$ of quantum correlation is expected to possess the following properties \cite{2014arXiv1411.3208S}. 
	\begin{enumerate}
		\item 
			$\mathcal{M}$ is non-negative.
		\item
			$\mathcal{M}$ is zero for classically correlated states.
		\item
			$\mathcal{M}$ in invariant under local unitary transformation.
	\end{enumerate}

	Setting $\mathcal{M}(\rho(G))=\mathcal{QD}(G)$ for any graph $G$ on $mn$ vertices, we have the following observations. First note that $\mathcal{QD}(G)\geq 0$ by default.  Further, by definition of $\mathcal{QD}(G)$ it is zero for classically correlated states with zero quantum discord. Finally, a general picture of graph theoretic analogue of local unitary transformations on $\rho(G),$ is not yet known, see \cite{dutta2016graph}. In fact, it is a hard problem to characterize local unitary operations which transform a given $\rho(G)$ in to an another state $\rho(H)$ for some other graph $H.$ However, we show in the next theorem that $\mathcal{QD}(G)$ is a promising measure for quantum correlation as it is invariant under local unitary operators of the form $P_1\otimes P_2$, where $P_i$, $i=1,2$ are permutation matrices, which are special unitary matrices. 

	Let the permutation matrices $P_1$ and $P_2$ act on the Hilbert spaces $\mathcal{H}^{(A)}$, and $\mathcal{H}^{(B)}$, respectively. Hence, $P_1 \otimes P_2 = (P_1 \otimes I)(I \otimes P_2)$ is a local unitary operator acting on $\mathcal{H}^{(A)} \otimes \mathcal{H}^{(B)}$. Then we have the following theorem:
	
	\begin{theorem}
		Let $P_2$ be a permutation matrix acting on the Hilbart space $\mathcal{H}^{(B)}$. Let $\rho(G)$ be a quantum state with zero discord in the bipartite system $\mathcal{H}^{(A)} \otimes \mathcal{H}^{(B)}$. Consider a graph $H$, such that, $A(H) = (I \otimes P_2)^t A(G) (I \otimes P_2)$. Then, $\rho(H)$ and $\rho(G)$ have equal quanutm discord.
	\end{theorem}

	\begin{proof} Note that 
	\begin{eqnarray}
		&& (I \otimes P_2)^t \rho(G) (I \otimes P_2) \nonumber \\ && = (I \otimes P_2^t) \frac{1}{d}\begin{bmatrix} D_1 +s A_{11} & s A_{12} & \dots & s A_{1m} \\ s A_{21} & D_2 +s A_{22} & \dots & s A_{2m} \\ \vdots & \vdots & \ddots & \vdots \\ s A_{m1} & s A_{m2} & \dots & D_m +s A_{mm}\end{bmatrix}(I \otimes P_2) \nonumber\\ && = \frac{1}{d}\begin{bmatrix} P_2^t(D_1 +s A_{11})P_2 & s P_2^tA_{12}P_2 & \dots & s P_2^tA_{1m}P_2 \\ s P_2^t A_{21} P_2 & P_2^t (D_2 +s A_{22}) P_2 & \dots & s P_2^t A_{2m} P_2 \\ \vdots & \vdots & \ddots & \vdots \\ s P_2^t A_{m1} P_2 & s P_2^t A_{m2} P_2 & \dots & P_2^t (D_m +s A_{mm}) P_2 \end{bmatrix}.\nonumber
	\end{eqnarray}

	Recall the subgraphs $\langle C_\mu \rangle$ and $\langle C_\mu, C_\nu \rangle$ in $G$ and the fact that graph isomorphisms are represented by permutation matrices. Hence, the above equation can be interpreted as a graph isomorphism operation. The adjacency matrix of the new subgraph corresponding to $\langle C_\mu \rangle$, and $\langle C_\mu, C_\nu \rangle$ are given by $P_2^t A_{\mu \mu} P_2$, and $\begin{bmatrix} 0 & P_2^t A_{\mu \nu}P_2 \\ P_2^t A^t_{\mu \nu}P_2 & 0 \end{bmatrix}$, respectively. Note that, the permutation matrix $P_2$ does not switch one vertex of $C_\mu$ to another vertex of $C_\nu$ when $\mu\neq \nu$ but only changes the labeling of vertices of $C\mu, 1\leq \mu\leq m.$ Thus, the normality and commutativity conditions hold as earlier in the new graph. Thus, $\rho(H)$ and $\rho(G)$ have equal quanutm discord.
	\end{proof} 

	The Werner state \cite{werner1989quantum} is a class of quantum states, important in quantum information processing. A Werner state is represented by,
	\begin{equation}\label{werner}
	\rho_{x,d} = \frac{d - x}{d^3 - d}I + \frac{xd - 1}{d^3 - d}F,
	\end{equation}
	where $F = \sum_{i,j}^d \ket{i}\bra{j} \otimes \ket{j}\bra{i}$, $x \in [0, 1]$ and $d$ is the dimention of the individual subsystems. Note that, $\rho_{x,d}$ is a symmetric matrix of order $d^2$. It can be shown that these states has non-zero quantum discord even though some of them are separable \cite{li2007total}. 
	
	\begin{example}
		We may represent $\rho_{1,3}$, and $\rho_{1,4}$ as a simple graph having $9$ and $16$ vertices in the figure \ref{werner_graph}.
		\begin{figure}
			\begin{center}
				\begin{subfigure}[b]{0.4\textwidth}
					\begin{tikzpicture}
					\node at (0,4) {$\bullet_{1,1}$};
					\node at (2,4) {$\bullet_{1,2}$};
					\node at (4,4) {$\bullet_{1,3}$};
					\node at (0,2) {$\bullet_{2,1}$};
					\node at (2,2) {$\bullet_{2,2}$};
					\node at (4,2) {$\bullet_{2,3}$};
					\node at (0,0) {$\bullet_{3,1}$};
					\node at (2,0) {$\bullet_{3,2}$};
					\node at (4,0) {$\bullet_{3,3}$};
					
					\draw {[rounded corners] (-.24,4) -- (.01,4.25) -- (-.24, 4.5) -- (-.49, 4.25) -- (-.24,4)};
					
					\draw {[rounded corners] (1.76,2) -- (2.01,2.25) -- (1.76, 2.5) -- (1.51, 2.25) -- (1.76,2)};
					
					\draw {[rounded corners] (3.76,0) -- (4.01,0.25) -- (3.76, 0.5) -- (3.51, 0.25) -- (3.76,0)};
					
					\draw {[rounded corners] (-.24,2) -- (.01,2.25) -- (-.24, 2.5) -- (-.49, 2.25) -- (-.24,2)};
					
					\draw {[rounded corners] (-.24,0) -- (.01,0.25) -- (-.24, 0.5) -- (-.49, 0.25) -- (-.24,0)};
					
					\draw {[rounded corners] (1.76,4) -- (2.01,4.25) -- (1.76, 4.5) -- (1.51, 4.25) -- (1.76,4)};
					
					\draw {[rounded corners] (1.76,0) -- (2.01,0.25) -- (1.76, 0.5) -- (1.51, 0.25) -- (1.76,0)};
					
					\draw {[rounded corners] (3.76,4) -- (4.01,4.25) -- (3.76, 4.5) -- (3.51, 4.25) -- (3.76,4)};
					
					\draw {[rounded corners] (3.76,2) -- (4.01,2.25) -- (3.76, 2.5) -- (3.51, 2.25) -- (3.76,2)};
					
					\draw (-.24,2) -- (1.76,4);
					
					\draw (1.76,0) -- (3.76,2);
					
					\draw {[rounded corners] (-.24,0) .. controls(2.5,1.5) .. (3.76,4)};
					
					\end{tikzpicture}
					\caption{Graph for $\rho_{1,3}$}
				\end{subfigure}
				\hspace{2cm}
				\begin{subfigure}[b]{0.4\textwidth}
					\begin{tikzpicture}
					\node at (0,6) {$\bullet_{1,1}$};
					\node at (2,6) {$\bullet_{1,2}$};
					\node at (4,6) {$\bullet_{1,3}$};
					\node at (6,6) {$\bullet_{1,4}$};
					\node at (0,4) {$\bullet_{2,1}$};
					\node at (2,4) {$\bullet_{2,2}$};
					\node at (4,4) {$\bullet_{2,3}$};
					\node at (6,4) {$\bullet_{2,4}$};
					\node at (0,2) {$\bullet_{3,1}$};
					\node at (2,2) {$\bullet_{3,2}$};
					\node at (4,2) {$\bullet_{3,3}$};
					\node at (6,2) {$\bullet_{1,4}$};
					\node at (0,0) {$\bullet_{4,1}$};
					\node at (2,0) {$\bullet_{4,2}$};
					\node at (4,0) {$\bullet_{4,3}$};
					\node at (6,0) {$\bullet_{4,4}$};
					
					\draw {[rounded corners] (-.24,6) -- (.01,6.25) -- (-.24, 6.5) -- (-.49, 6.25) -- (-.24,6)};
					
					\draw {[rounded corners] (1.76,4) -- (2.01,4.25) -- (1.76, 4.5) -- (1.51, 4.25) -- (1.76,4)};
					
					\draw {[rounded corners] (3.76,2) -- (4.01,2.25) -- (3.76, 2.5) -- (3.51, 2.25) -- (3.76,2)};
					
					\draw {[rounded corners] (5.76,0) -- (6.01,0.25) -- (5.76, 0.5) -- (5.51, 0.25) -- (5.76,0)};
					
					\draw {[rounded corners] (-.24,4) -- (.01,4.25) -- (-.24, 4.5) -- (-.49, 4.25) -- (-.24,4)};
					
					\draw {[rounded corners] (1.76,6) -- (2.01,6.25) -- (1.76, 6.5) -- (1.51, 6.25) -- (1.76,6)};
					
					\draw {[rounded corners] (1.76,2) -- (2.01,2.25) -- (1.76, 2.5) -- (1.51, 2.25) -- (1.76,2)};
					
					\draw {[rounded corners] (3.76,0) -- (4.01,0.25) -- (3.76, 0.5) -- (3.51, 0.25) -- (3.76,0)};
					
					\draw {[rounded corners] (-.24,2) -- (.01,2.25) -- (-.24, 2.5) -- (-.49, 2.25) -- (-.24,2)};
					
					\draw {[rounded corners] (-.24,0) -- (.01,0.25) -- (-.24, 0.5) -- (-.49, 0.25) -- (-.24,0)};
					
					\draw {[rounded corners] (1.76,0) -- (2.01,0.25) -- (1.76, 0.5) -- (1.51, 0.25) -- (1.76,0)};
					
					\draw {[rounded corners] (3.76,4) -- (4.01,4.25) -- (3.76, 4.5) -- (3.51, 4.25) -- (3.76,4)};
					
					\draw {[rounded corners] (5.76,2) -- (6.01,2.25) -- (5.76, 2.5) -- (5.51, 2.25) -- (5.76,2)};
					
					\draw {[rounded corners] (5.76,4) -- (6.01,4.25) -- (5.76, 4.5) -- (5.51, 4.25) -- (5.76,4)};
					
					\draw {[rounded corners] (5.76,6) -- (6.01,6.25) -- (5.76, 6.5) -- (5.51, 6.25) -- (5.76,6)};
					
					\draw {[rounded corners] (3.76,6) -- (4.01,6.25) -- (3.76, 6.5) -- (3.51, 6.25) -- (3.76,6)};
					
					\draw (-.24,4) -- (1.76,6);
					
					\draw (1.76,2) -- (3.76,4);
					
					\draw {[rounded corners] (-.24,2) .. controls(2.5,3.5) .. (3.76,6)};
					
					\draw {[rounded corners] (1.76,0) .. controls(4,1.5) .. (5.76,4)};
					
					\draw (3.76,0) -- (5.76,2);
					
					\draw {[rounded corners] (-.24,0) .. controls(4,3) .. (5.76,6)};
					
					\end{tikzpicture}
					\caption{Graph for $\rho_{1,4}$}
				\end{subfigure}
				\caption{Graphs of the Werner states}
				\label{werner_graph}
			\end{center}
		\end{figure}
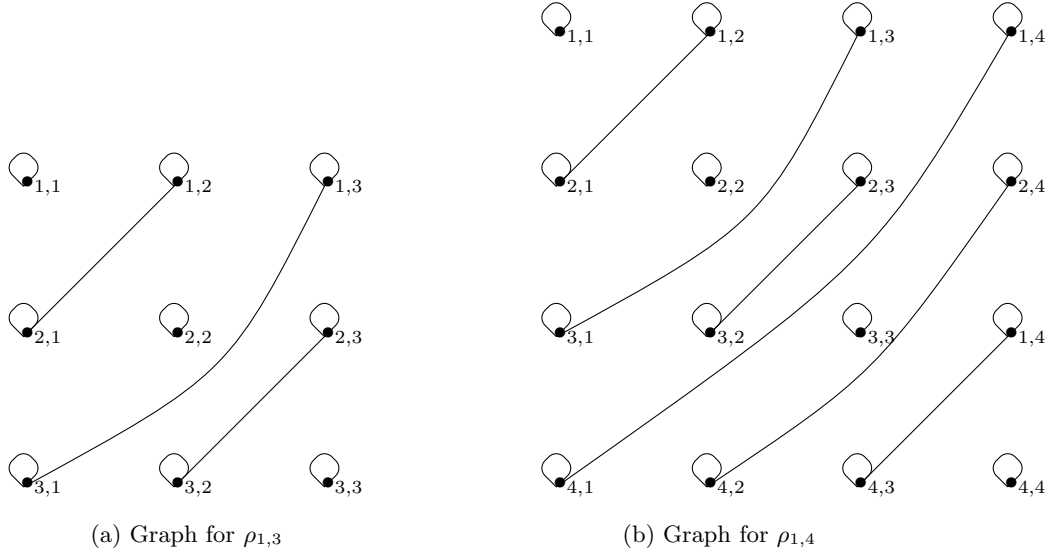
	\end{example}
	
	\begin{theorem}
		Every Werner state has non-zero discord.
	\end{theorem}
	
	\begin{proof} 
		Consider the subgraph $\langle C_\mu, C_\nu \rangle$ for any $\mu$ and $\nu$, with $\mu \neq \nu$. Using the lemma \ref{normality}, we may conclude that $A_{\mu, \nu}$ is not a normal matrix. Thus every Werner state has a non-zero discord. 
	\end{proof}
	
	A detailed study of quantum discord of states represented with weighted graphs will be presented in an upcoming work \cite{next}.

	\section{Graph theoretic zero quantum discord states}

	As discussed above, we can generate zero quantum discord bipartite states of dimension $m\times n$ arising from graphs for any $m, n$ by generating graphs $G$ for which $\mathcal{QD}(G)=0.$ In fact states arising from a graph $G$ have zero quantum discord if and only if $\mathcal{QD}(G)=0.$ In this section we determine certain standard graphs which have always zero quantum discord for any labelling on the vertices.

	First we have the following theorem for complete graphs on $N=mn$ vertices which are graphs in which any pair of distinct vertices are adjacent.

	\begin{theorem}\label{Thm:com}
		Let $G$ be a complete graph on $N=mn$ vertices. Then the states $\rho(G)$ have zero quantum discord, that is, $\mathcal{QD}(G)=0$.
	\end{theorem}
	
	\begin{proof}
		As $G$ is a complete graph, degree of every vertex is $N -1$. Thus total degree of $G$ is $d = N(N - 1)$. Then the  blocks of $\rho(G)$ are given by
		\begin{equation} 
			\rho_{\mu \nu} = \begin{cases} \frac{\pm 1}{d}A_{\mu \mu} = \frac{1}{d}\left[(N - 1) I_n + J_n - I_n\right] = \frac{1}{d}[(N - 2)I_n + J_n] & ~\text{for}~ \mu = \nu \\ \frac{\pm 1}{d}A_{\mu \nu} = \frac{\pm 1}{d}J_n & ~\text{for}~ \mu \neq \nu \end{cases}
		\end{equation} where  $1\leq \mu,\nu \leq m$ and $J_n$ is the all-one matrix of order $n.$
		Note that all the blocks $\rho_{\mu \nu}$ are normal and commute pairwise. Hence the desired result follows.
	\end{proof}
	
	We conclude from Theorem \ref{Thm:com} that for any $m$ and $n$ there is a graph $G$ of order $mn$ for which the corresponding bipartite states of dimension $m\times n$ have zero quanum discord.  In the next theorem we prove that given any $n,$ the complete bipartite graphs $G$ on $2n$ vertices, that is, $G$ consists of two clusters $\langle C_\mu\rangle, \mu=1,2$ with $n$ vertices such that no two vertices of $C_\mu$ for a fixed $\mu$ are adjacent and all pairs of vertices $u\in C_1, v\in C_2,$ $(u, v)\in E(G),$ provides zero quantum discord bipartite states $\rho(G)$ of dimension $2\times n.$ 

	\begin{theorem}\label{Thm:cb}
		Let $G$ be a complete bipartite graph on $2n$ vertices with two clusters, each on $n$ vertices. Then $\mathcal{QD}(G)=0,$ that is, $\rho(G)$ have zero quantum discord. 
	\end{theorem}
	\begin{proof}
		Let $C_1 = \{v_{1,1}, v_{1,2}, \dots v_{1,n}\}$, $C_2 = \{v_{2,1}, v_{2,2}, \dots v_{2,n}\}$ be the bipartition of $G$. Then
		$$\rho(G) = \frac{1}{2n}\begin{bmatrix} nI_n & s J_n \\ s J_n & nI_n \end{bmatrix}.$$
		It is easy to verify that all the block matrices commute with each other and they are normal matrices. Hence $\mathcal{QD}(G)=0.$ 
	\end{proof}

	Next, we provide an example of two isomorphic graphs such that for one the corresponding bipartite states have zero quantum discord and for the other, the corresponding states have non-zero quantum discord. Thus, the following example establishes that quantum discord is not invariant under graph isomorphism, hence depends on labeling of the vertices.  

	\begin{example}
		In the figure \ref{bipartite_graphs}, there are two isomorphic complete bipartite graphs $G$ and $H$ with vertex set $V = \{1,2,3,4,5,6\}$. It consists of two clusters $C_1 = \{1,2,3\}$ and $C_2 = \{4,5,6\}$. A simple calculation shows that $\mathcal{QD}(H)\neq 0$ and from Theorem \ref{Thm:cb}, $\mathcal{QD}(G)=0.$ It is interesting to observe that $H$ is a $3$-regular graph which confirms that regular graphs need not represent states with zero quantum discord.
		\begin{figure}
			\begin{subfigure}[b]{0.4\textwidth}
				\begin{tikzpicture}
					\node (4) at (0,0) {$\bullet_4$};
					\node (5) at (2,0) {$\bullet_5$};
					\node (6) at (4,0) {$\bullet_6$};
					\node (1) at (0,2) {$\bullet_1$};
					\node (2) at (2,2) {$\bullet_2$};
					\node (3) at (4,2) {$\bullet_3$};
					
					\draw (1) -- (4);
					\draw (1) -- (5);
					\draw (1) -- (6);
					\draw (2) -- (4);
					\draw (2) -- (5);
					\draw (2) -- (6);
					\draw (3) -- (4);
					\draw (3) -- (5);
					\draw (3) -- (6);
				\end{tikzpicture}
				\caption{Graph $G$}
			\end{subfigure}
			\hspace{2cm}
			\begin{subfigure}[b]{0.4\textwidth}
				\begin{tikzpicture}
					\node (4) at (0,0) {$\bullet_4$};
					\node (5) at (2,0) {$\bullet_5$};
					\node (6) at (4,0) {$\bullet_6$};
					\node (1) at (0,2) {$\bullet_1$};
					\node (2) at (2,2) {$\bullet_2$};
					\node (3) at (4,2) {$\bullet_3$};
					
					\draw (1) -- (2);
					\draw (1) -- (6);
					\draw (2) -- (1);
					\draw (2) -- (4);
					\draw (2) -- (5);
					\draw (3) -- (4);
					\draw (3) -- (5);
					\draw (1) .. controls(2, 2.5) .. (3);
					\draw (4) .. controls(2, -.5) .. (6);
				\end{tikzpicture}
				\caption{Graph $H$}
			\end{subfigure}
			\caption{Isomorphic complete bipartite graphs with different quantum discords.}
			\label{bipartite_graphs}
		\end{figure}
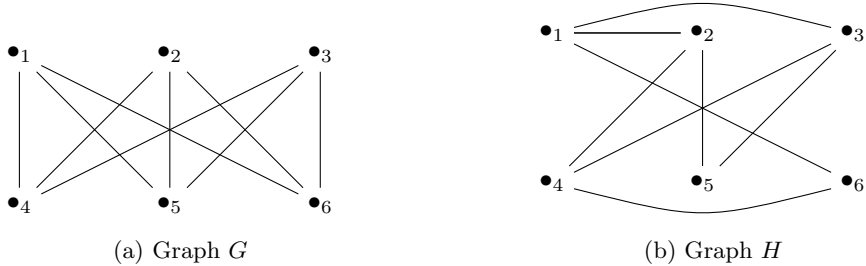
	\end{example}

	Now we consider partially symmetric graphs which were introduced in \cite{dutta2016bipartite} and show that states corresponding to  bipartite partially symmetric regular graphs have always zero quantum discord. First we recall the following definition from \cite{dutta2016bipartite}. 
	
	\begin{definition}{\bf Partially symmetric graph:}
		A graph $G$ with clusters $C_1, C_2, \dots C_m$ is called a partially symmetric graph if the edge $(v_{\mu i}, v_{\nu j}) \in E(G)$ indicates, $(v_{\nu i}, v_{\mu j}) \in E(G)$.
	\end{definition}
	
	Note from the definition that every block of the adjacency matrix of a partially symmetric graph is a symmetric matrix. Then we have the following theorem.
	
	\begin{theorem}
		Every bipartite partially symmetric regular graph has zero quantum discord.
	\end{theorem}
	\begin{proof}
		Let $G$ be a partially symmetric regular bipartite graph. Then
		$$\rho(G) = \begin{bmatrix} rI_n & s A_n \\ s A_n & rI_n \end{bmatrix},$$
		where $r$ is the regularity of the graph and $s\in\{1,-1\}$. Since $A_n$ is a symmetric matrix, it is normal. Besides, all these block matrices commute with each other. Hence, $\mathcal{QD}(G)=0.$
	\end{proof}

	\begin{theorem}
		Let $G = \langle C_\mu, C_\nu \rangle$ be a regular graph satisfying the condition of Theorem \ref{normality}. Then the states $\rho(G)$  have zero quantum discord.
	\end{theorem}

	\begin{proof}
		The proof is similar to the last theorem. Here the matrix $A_n$ is a normal matrix instead of symmetric matrix. Indeed the quantum states corresponding to $G$ are given by
		\begin{equation}
			\rho(G) = \frac{1}{2nr}\begin{bmatrix} rI_n & s A_n \\ s A^t_n & rI_n \end{bmatrix} = \frac{1}{2n}\begin{bmatrix} I_n & s \frac{1}{r} A_n \\ s \frac{1}{r} A^t_n & I_n \end{bmatrix}, ~\text{where}~ s\in\{1,-1\}.
		\end{equation}
	\end{proof} 

	Consider the matrix $\begin{bmatrix} rI_n & A_n \\ A^t_n & rI_n \end{bmatrix}$ in the above equation by putting $s = 1$. Every row and column has equal sum $2r$. The matrix $\frac{1}{2r}\begin{bmatrix} rI_n & A_n \\ A^t_n & rI_n \end{bmatrix}$ is an example of a doubly-stochastic matrix. Doubly stochastic matrices are widely used in different branches of science \cite{horn2012matrix}.
	
	Finally as it is well known that there are separable quantum states with non-zero quantum discord, in the following example we confirm the same also for the states arising from graphs.
	\begin{example}
		Consider the bipartite partially symmetric graph $G$ representing a separable two-qubit mixed state.
		$$
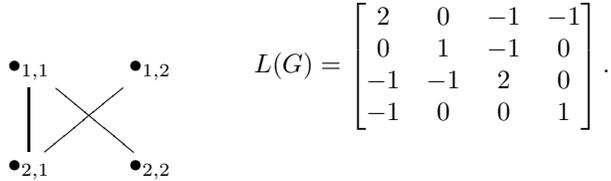
\xymatrix{\bullet_{1,1} \ar@{-}[d] \ar@{-}[dr] & \bullet_{1,2} \ar@{-}[dl] \\ \bullet_{2,1} & \bullet_{2,2}} \hspace{1cm} L(G) = \begin{bmatrix}  2&  0& -1& -1 \\ 0&  1& -1&  0 \\ -1& -1&  2&  0 \\ -1&  0&  0&  1\end{bmatrix}.$$
		Note that $\begin{bmatrix} 2& 0 \\ 0 & 1\end{bmatrix}$ and $\begin{bmatrix} -1& -1 \\ -1&   0\end{bmatrix}$ do not commute, hence $\mathcal{QD}(G)\neq 0$ although $\rho(G)=\dfrac{1}{5}L(G)$ represents a $2$-qubit separable state. 
	\end{example}

	\section{Conclusions and open problems}
	
	This work is important from the perspective of mathematics and theoretical quantum information. Calculating the exact amount of quantum discord is a computationally formidable task. Here, we derive graph theoretic criteria for normality and commutativity of binary matrices. The blocks of a density matrix of a zero discord state are normal and commuting.  We apply combinatorial tools to find out graph theoretic criterion for zero quantum discord. Further, we propose a graph theoretic measure of discord. This work initiates a number of directions for future research.
	
	\begin{enumerate}
		\item
			Given a positive integer $n$, calculate the exact number of binary normal matrices. There is no general formula for this problem till date, although some lower bound exists. This work provides a graph theoretic visualization to the structure of binary normal matrices, which may be useful for solving this problem.

		\item
			The Werner states play an important role in quantum information theory. It was proved in \cite{li2007total} that they have non-zero quantum discord. These states can be represented by weighted graphs and will be considered in a forthcoming work.
	\end{enumerate}

	\section*{Acknowledgment}
	This work was partially supported by the project \textit{“Graph theoretical aspects in quantum information processing”} [Grant No. 25(0210)/13/EMR-II] funded by Council of Scientific and Industrial Research, New Delhi. S.D. is grateful to the Ministry of Human Resource Development, Government of India, for a doctoral fellowship. This work may be a part of his doctoral thesis.

	\bibliographystyle{unsrt}

\end{document}